\documentclass{article}
\usepackage{graphicx} 
\usepackage{amsmath}
\usepackage{amsfonts}
\usepackage{amssymb}
\usepackage{amsthm}
\usepackage{hyperref}
\usepackage{mathtools}
\usepackage{xcolor}
\usepackage{tikz}
\usepackage{anyfontsize}
\usetikzlibrary{arrows.meta}
\usetikzlibrary{positioning,fit,calc}
\usetikzlibrary{decorations.pathmorphing}
\usetikzlibrary{bending}

\title{On certain combinatorial expressions of TASEP transition
probabilities}
\author{Lorenzo Vito Dal Zovo}
\date{Politecnico di Torino\\
    lorenzo.dalzovo@polito.it}
\begin{document}
\newtheorem{theorem}{Theorem}
\newtheorem{corollary}{Corollary}
\newtheorem{lemma}{Lemma}

\maketitle

\begin{abstract}

We study combinatorial structures arising from finite-time transition probabilities of the Totally Asymmetric Simple Exclusion Process with open boundary conditions. While much of the existing combinatorial theory regarding the TASEP concerns the steady-state distribution, we focus instead on the transient dynamics. We first show that the enumeration of transition sequences between two configurations of the open TASEP is equivalent to the enumeration of standard Young tableaux of a family of non-classical shapes which have been of recent interest in the combinatorial literature. This extends to the open-boundary setting the correspondence between the TASEP with periodic boundaries and cylindric tableaux.

We then introduce a family of tableau-like objects associated with Young diagrams in which repetitions of cells are allowed, subject to the partial order induced by the diagram. For each diagram, we collect the numbers of these objects into an exponential generating function. We prove that the entries of the homogeneous open TASEP transition matrix can be expressed as signed sums of such generating functions over suitable families of diagrams. This gives a combinatorial and order-theoretic interpretation of finite-time transition probabilities for the open TASEP, analogous to the combinatorial mappings known for steady-state probabilities.

\end{abstract}

\section{Introduction}
This paper focuses on the Totally Asymmetric Simple Exclusion Process (TASEP) with open boundaries. The TASEP describes the dynamics of hard-core particles hopping on a one-dimensional lattice. Particles can only enter from the left, hop to the right if the target site is free, and exit from the right end of the lattice. In the inhomogeneous case, each site is equipped with a given rate (including a fictional zero site accounting for particle ingress), while in the homogeneous case all rates are taken to be equal, except possibly for the entry and exit rates.

The homogeneous open  TASEP is a paradigmatic model for a number of reasons. Notably, it is an integrable model, in the sense that its steady-state distribution can be computed explicitly, it has a direct biological interpretation and it experiences boundary-driven phase transitions \cite{Blythe_2007}. 

The homogeneous TASEP steady-state distribution over its microstates can be computed through the Matrix Factorization Ansatz (MFA). The structure of this solution has led to a number of combinatorial mappings, motivated by the presence of ubiquitous number sequences within it \cite{Wood_2020}.

As such, previous works have focused on the steady-state distribution, rather than on transient probabilities, for which closed-form results are known only in the periodic case and in the case of an infinite lattice \cite{PhysRevLett.91.050601, schutz1997exact}.

The purpose of the present work is to provide a combinatorial description of certain objects related to the transient dynamics of the homogeneous TASEP, that is finite-time transition probabilities and sequences of transitions. 

In Section 2 we introduce the TASEP rigorously with sitewise inhomogeneous parameters and define all combinatorial objects of interest. In particular, we introduce a multiset generalization of the number of standard tableaux associated with a Young diagram. In particular, we view the diagram as a partially ordered set (poset) and allow for repetitions of its elements, maintaining the same order relations between different elements.

In Section 3.1 we discuss the problem of enumerating all possible walks of a given length from one configuration to another, similarly to what was done for the periodic TASEP \cite{elizalde2022walks}. In doing so we show how this problem coincides with that of enumerating the number of standard Young tableaux for a given family of Young diagrams, which have been of recent interest in the combinatorial literature \cite{Sun2015EnumerationOS,VeraLopez2017,EnumTableaux}.

In Section 3.2, we turn to finite-time transition probabilities. We express the entries of the transition probability matrix as signed sums of exponential generating functions associated with the multiset generalization of Young tableaux introduced above. These objects are distinct from other multiset-based generalizations of partially ordered structures, such as the partially ordered multisets, or pomsets, introduced to model concurrency \cite{Pratt1986ModelingCW}.

\section{Preliminaries}

\subsection{The TASEP with open boundaries}

The relative simplicity of the steady-state probability distribution, which can be given in closed form even for finite lattice sizes, has led to a vast body of work related to combinatorial interpretations of such results and techniques for the extrapolation of new quantities. The author refers the reader to \cite{Wood_2020} for further information. 

The goal of the present work is to discuss combinatorial mappings related to the computation of the exact finite-time probabilities for the TASEP on a finite segment, which are independent of the parameter choices. The Matrix Factorization Ansatz does not provide information on transient states and their probabilities, thus the techniques and results from \cite{Wood_2020} cannot be applied. 

The TASEP on a segment of length $N$ can be described by a continuous-time Markov chain (CTMC) on the state space $X_N=\{0,1\}^N$, where zero represents an empty site and one an occupied one. For any $x\in X_N$ we introduce the index map
\begin{equation*}
    \operatorname{idx}\colon \,X_N\to[2^N],\qquad
    \operatorname{idx}(x)=1+\sum_{k=1}^N x_k \cdot 2^{N-k},
\end{equation*}
where $\operatorname{idx}(x)\in\left[2^N\right]$.
We recall the definition of Kronecker product of two generic matrices $A\in\mathbb{R}^{p,\,q}$, $B\in\mathbb{R}^{m,\,n}$ as the matrix $C\in\mathbb{R} ^{mp,\,qn}$ defined as:
\[
C=A \otimes B \coloneq\left(\begin{array}{ccc}
        a_{11} B & \ldots & a_{1 n} B \\
        \vdots & \ddots & \vdots \\
        a_{m 1} B & \cdots & a_{m n} B
        \end{array}\right).
\]
We also introduce the notation $A^{\otimes N}\coloneq\underbrace{A\otimes A \otimes \cdots\otimes A}_{N\text{ times}}$ and denote by $I$ the two by two identity matrix.
One can then construct the generator:
\begin{align*}
    H &= \sum_{k=1}^{N+1} h_k, \quad\text{where:} \\
     h_1 &= \theta_1\big(a \otimes I^{\otimes(N-1)}\big),\\
    h_j &= \theta_j\big(I^{\otimes(j-2)} \otimes c \otimes I^{\otimes(N-j)}\big) \quad \text{for} \quad j=2,\dots,N, \\
    h_{N+1}&= \theta_{N+1}\big(I^{\otimes(N-1)} \otimes b\big),
\end{align*}
and where $\boldsymbol{\theta}=(\theta_i)_{i=1}^{N+1}$ is the set of transition rates (each rate is set to one in the homogeneous case), $\{h_i\}$ is the set of site-wise generators and:
\[    a=\begin{pmatrix}
    -1 & 0 \\
    1 & 0
    \end{pmatrix}, \quad 
    b=\begin{pmatrix}
    0 & 1 \\
    0 & -1
    \end{pmatrix}, \quad 
    c=\begin{pmatrix}
    0 & 0 & 0 & 0 \\
    0 & 0 & 1 & 0 \\
    0 & 0 & -1 & 0 \\
    0 & 0 & 0 & 0
    \end{pmatrix},
\]
are the local generators.
 Each $h_k$ is a local infinitesimal generator describing a nearest-neighbor or boundary transition.

The local generator matrices are known to satisfy algebraic relations in the partially asymmetric case \cite{schutz2001}, which reduce in the case of the TASEP to:
\begin{equation}
\label{algebraic rel}
\begin{array}{rlr}
h_i h_{i \pm 1} h_i & =0 & \forall i \in\{2, \ldots, N\} \\
h_i^m & =(-1)^{m-1}\theta_i^{m-1} h_i & \forall i \in\{1, \ldots, N+1\} \\
{\left[h_i, h_j\right]} & =0 & \forall i,j \in [N+1], \ \text{s.t. } |i-j|>1.
\end{array}
\end{equation}

Any CTMC can be interpreted as a random walk on a graph. The system Hamiltonian can thus be interpreted as the Laplacian of a graph $G_N$ whose node set is the configuration space $X_N$. A node is connected to another by a directed edge if a single transition between the two configurations is possible, with the corresponding transition rate as edge weight. One then has:
\begin{equation*}
    H=A-D,
\end{equation*}
where $D$ is the out-degree matrix and $A$ is the transition rate matrix. The number of walks from configuration $x$ to configuration $y$ will be denoted $W(x,y)$, and the number of such walks of prescribed length $n$ as $W(x,y,n)$.

The local matrices can be rewritten as:
\begin{equation*}
    a=E_{21}-E_{11}\,, 
    \quad
    h^{(r)}=E_{12}-E_{22}\,, 
    \quad\text{and}\quad
    h= E_{12} \otimes E_{21}-E_{22}\otimes E_{11}\,,
\end{equation*}
where $E_{ij}$ are the standard basis matrices of $\mathbb{R}^{2\times 2}$, satisfying $E_{ij} E_{kl}=\delta_{jk} E_{il}$.

For later use, we denote by
\begin{equation*}
    P(t)\coloneqq e^{Ht}
\end{equation*}
the finite-time transition probability matrix. Its entry \(P_{ij}(t)\) gives the probability of being in configuration indexed by \(i\) at time \(t\) when the process starts from the configuration indexed by \(j\). Equivalently, if \(p(t)\) denotes the column vector of state probabilities, then \(p(t)=P(t)p(0)\).

Since \(P(t)\) is the exponential of the generator, one has

\begin{equation*}
    P(t)=\sum_{n=0}^{\infty}\frac{t^n}{n!}H^n.
\end{equation*}

As the local generators do not all mutually commute, the powers $H^n$ can naturally be seen as a sum over the set of all possible words of length $n$ in the alphabet of the local generators $\Sigma=\{h_1,\cdots,h_{N+1}\}$, which we denote $\Sigma^n$.
From this perspective, a natural problem is to understand which families of words in the local generators contribute to a fixed matrix entry of \(P(t)\), and how those contributions can be grouped according to the partial orders naturally induced by the TASEP dynamics. This is the point of view adopted in Section 3.2, where the entries of \(P(t)\) will be written in terms of exponential generating functions associated with suitable diagrams.

\subsection{ Standard Young tableaux of (non-)classical shape}
\label{comb_prel}
We start by recalling some notions concerning partially ordered sets (posets) and we refer the reader to \cite{EnumComb2} for further information.

A poset is an object of the form $(P,\preceq)$, where $P$ is a set and $\preceq$ a partial order relation on $P$. We will by abuse of notation refer to both the poset and its underlying set as $P$. Let $|P|=n$, an order preserving bijection of the form $f\colon P\to[n]$ is called a linear extension of $P$ and we call $e(P)$ the number of said linear extensions. 

We now give the definition of a Young diagram $D$ and its associated set of standard Young tableaux $SYT(D)$.
We call a diagram any subset $D\subset\mathbb{Z}^2$ endowed with the natural partial order:
\begin{equation}
\label{partial}
(i, j) \preceq\left(i^{\prime}, j^{\prime}\right) \Longleftrightarrow i \leq i^{\prime} \text { and } j \leq j^{\prime}.
\end{equation}
Points in a diagram are referred to as cells, while a cell that is at the bottom of its column and is the rightmost of its row is said to be a corner.
Then, a standard Young tableau is an order preserving bijection $T\colon D\to [n]$, where $n=|D|$, that is:
\[
    c \leq_D c^{\prime} \Longrightarrow T(c) \leq T\left(c^{\prime}\right).
\]

Standard Young tableaux are invariant up to rigid translations in $\mathbb{Z}^2$ of the corresponding diagram, thus a common abuse of notation is to consider a diagram as an equivalence class  rather than a subset of $\mathbb{Z}^2$ .
A standard Young tableau can be visualized as a "filling" of the diagram $D$, where each element in $D$ is viewed as a cell containing a number from $1$ to $|D|$ such that the numbers increase column-wise 
and row-wise. We refer to the set of Young tableaux associated with $D$ as $SYT(D)$ and its cardinality as $f^D=|SYT(D)|$.
It is clear that a Young diagram can be understood as a poset and $f^D$ as the number of its linear extensions. \\

The enumeration of $f^D$ for certain families of $D$ is a classic problem in enumerative combinatorics and an active area of research \cite{EnumTableaux}.

We introduce a family of diagrams formed by $n$ shifted strips of length $N$ as the set of diagrams of the form:
\[
S_{N,n}=\bigcup_{j \in [n]}\left(X+j\Delta\right),
\]
where $n\in\mathbb{N}$, $\Delta=(1,-1)$ and $X=\left\{(0,0),(0,1),\cdots,(0,N)\right\}$. These diagrams and the enumeration of $f^{S_{N,n}}$ have been of recent interest in the combinatorial literature, where it was proven that the sequence $a_n=f^{S_{N,n}}$ obeys a constant coefficient recurrence relation for all $N\in\mathbb{N}$. This was first proven for $N$ up to $5$ through integral methods in \cite{Sun2015EnumerationOS}, then for all $N$ in \cite{VeraLopez2017}. The latter proof is constructive, in the sense that it provides a direct construction of the recurrence polynomial. 

We are also interested in the infinite shifted tableaux of length $N$:
\begin{equation}
\label{inifinite shape}
Z_{N}=\bigcup_{j \in \mathbb{N}}\left(X+j\Delta\right).
\end{equation}
We now define the set $D_N$ of sub-diagrams of $Z_N$ obeying the following conditions:
\begin{itemize}
    \item if $D\in\ D_N$ then $(0,0)\in D$;
    \item if $D\in D_N$ then $D$ is row and column justified;
    \item if $D\in D_N$ then $D$ is path connected \textit{and} line-convex.
\end{itemize}
The \textit{Young graph} $\mathrm{Y}$ associated with $Z_{N}$ is the graph having node set $D_N$, where a node $x$ is connected by an edge to a node $y$ whenever it is possible to obtain $y$ from $x$ by adding a corner to it. 

A diagram $D\in D_{N+1}$ can be uniquely determined by its cardinality and by the path in $\mathbb{Z}^2$ defined by starting in its bottom-left cell and moving rightward whenever possible and upwards otherwise, until neither is possible. We call such a path the outline of $D$, written $h(D)$, and, noting it must consist of $N-1$ steps, we identify it with an element of $\{\uparrow,\to\}^{N-1}$.

We finally define a set of diagrams which will be useful later on:
\[
D^N_{\eta,\xi}=\left\{D\subset Z_N \ | \ D=A\setminus B, \text{ where }A,B\in D_N, \ B\subseteq A, \ h(A)=\eta, \ h(B)=\xi   \right\}.
\]
The set $D^N_{\eta,\xi}$ contains all diagrams "delimited" by the outlines $\eta, \ \xi$, which differ from one another by their size, as can be seen in Fig. \ref{Fig: diagrams}.

The results concerning shifted strips have recently been extended to diagrams of shape $D=D_1\setminus D_2$ when $h(D_1)=h(D_2)$. In that case the outline defines a path-connected diagram $Y$ containing no $2\times2$ squares, known in the literature as a zigzag shape \cite{EnumTableaux}.
It is easy to see then that $D$ can be written, for some $n$, as:
\[
D_{X,n}=\bigcup_{j \in [n]}\left(X+j\Delta\right),
\]
where $X$ is a zigzag shape and $\Delta$ is as above.

In \cite{periodiPPartitions} the results presented for the case of shifted strips were extended to standard tableaux of the form $D_{X,n}$, which we will refer to as shifted shapes. In particular it was shown that the sequence $a_n=f^{D_{X,n}}$ obeys a constant coefficient recurrence relation and its asymptotic behavior was described.

We now construct, for any Young diagram $D$, an integer sequence in a sense generalizing the number $f^D$. Let $n\in\mathbb{N}$ and consider a surjective labeling function $\ell\colon [n]\to D$. Let $F(D,n)$ be the set of all such labelings such that:
\begin{equation}
    \forall i,j\in[n ],\quad i<j\Rightarrow\ell(i)\preceq\ell(j)\text{ or}\ \ell(i)\text{ and } \ell(j)\text{ are not comparable},
\end{equation}
where $\preceq$ is the partial order defined in (\ref{partial}).
We define:
\begin{equation}
    \label{contatore parole}
    f^D(n)=|F(D,n)|
\end{equation}
where we remark that $f^D(|D|)=f^D$, while $f^D(n)=0$ if and only if $n<|D|$. 
Equivalently, \(f^D(n)\) counts the number of sequences of elements in $D$ of length $n$, allowing for repetitions, such that for every pair of distinct comparable elements $c,\,d\in D$ such that $c\preceq d$, then every instance of $c$ precedes any instance of $d$.


In this sense, the numbers \(f^D(n)\) extend the usual tableaux enumeration \(f^D\). In fact, when \(n=|D|\) one recovers standard Young tableaux, while for \(n>|D|\) repetitions are allowed in a way compatible with the order structure of the diagram.

Finally, we define the exponential generating function of $f^D(n)$ as:
\begin{equation}
\label{egf}
\mathcal{F}^D(t)=\sum_{n=0}^\infty f^D(n)\frac{t^n}{n!}.
\end{equation}

We note that $\mathcal{F}^D(z)$, with $z\in\mathbb{C}\,$ is not only a formal power series but an analytic function in the entirety of $\mathbb{C}$. To note this let $k\coloneq|D|$ and observe that:
\[
\lim_{n\to\infty} \left|\frac{f^D(n)}{n!}\right|^{1/n}\leq\lim_{n\to\infty} \left|\frac{n^k}{n!}\right|^{1/n}=0
\]
where the first inequality is due to the fact that $n^k$ counts all possible strings of length $n$ in $k$ letters with no restriction on their relative order and thus is an overestimate of $f^D(n)$. It follows that the first limit is equal to zero and, thanks to the Cauchy-Hadamard theorem, the radius of convergence is infinite.

\begin{figure}
    \centering
    
\begin{tikzpicture}[x=0.55cm,y=0.55cm,>=Latex,line join=round,line cap=round]

\tikzset{
  cell/.style={draw=black!55, line width=0.35pt},
  patharrow/.style={-{Latex[length=2.2mm,width=1.8mm]}, line width=0.9pt},
  dotstyle/.style={fill=black,draw=none},
}

\colorlet{panelcyan}{cyan!14}
\colorlet{fillA}{blue!60!cyan!30}

\def\RowLenA{20}
\def\RowLenB{20}
\def\RowLenC{20}
\def\RowLenD{20}
\def\RowLenE{20}

\def\ShiftA{0}
\def\ShiftB{1}
\def\ShiftC{2}
\def\ShiftD{3}
\def\ShiftE{4}

\begin{scope}[shift={(0,-0.5)}]

\node[anchor=east,font=\fontsize{11}{11}\selectfont] at (-0.8,0) {a)};

\foreach \x in {0,...,12}{
  \fill[panelcyan] (\x,0) rectangle ++(1,-1);
}
\foreach \x in {1,...,13}{
  \fill[panelcyan] (\x,-1) rectangle ++(1,-1);
}
\foreach \x in {2,...,13}{
  \fill[panelcyan] (\x,-2) rectangle ++(1,-1);
}
\foreach \x in {3,...,9}{
  \fill[panelcyan] (\x,-3) rectangle ++(1,-1);
}

\foreach \x in {0,...,12}{
  \draw[cell] (\x,0) rectangle ++(1,-1);
}
\foreach \x in {1,...,13}{
  \draw[cell] (\x,-1) rectangle ++(1,-1);
}
\foreach \x in {2,...,14}{
  \draw[cell] (\x,-2) rectangle ++(1,-1);
}
\foreach \x in {3,...,15}{
  \draw[cell] (\x,-3) rectangle ++(1,-1);
}
\foreach \x in {4,...,16}{
  \draw[cell] (\x,-4) rectangle ++(1,-1);
}

\foreach \x in {3.5,4.5,5.5,6.5,7.5,8.5}{
  \draw[patharrow] (\x,-3.5)--++(1.1,0);
}

\draw[patharrow] (9.5,-3.5)--++(0,1.2);

\foreach \x in {9.5,10.5,11.5,12.5}{
  \draw[patharrow] (\x,-2.5)--++(1.1,0);
}
\draw[patharrow] (13.5,-2.5)--++(0,1.2);

\node[anchor=west,font=\Large] at (14.2,-1.0) {$\xi$};

\fill[dotstyle] (5.4,-5.4) circle (1.1pt);
\fill[dotstyle] (5.7,-5.7) circle (1.1pt);
\fill[dotstyle] (6.0,-6.0) circle (1.1pt);
\fill[dotstyle] (10.4,-5.4) circle (1.1pt);
\fill[dotstyle] (10.7,-5.7) circle (1.1pt);
\fill[dotstyle] (11.0,-6.0) circle (1.1pt);
\fill[dotstyle] (17.4,-5.4) circle (1.1pt);
\fill[dotstyle] (17.7,-5.7) circle (1.1pt);
\fill[dotstyle] (18.0,-6.0) circle (1.1pt);

\end{scope}

\begin{scope}[shift={(0,-8.8)}]

\node[anchor=east,font=\fontsize{11}{11}\selectfont] at (-0.8,0) {b)};

\foreach \x in {6,...,13}{
  \fill[panelcyan] (\x,-1) rectangle ++(1,-1);
}
\foreach \x in {3,...,13}{
  \fill[panelcyan] (\x,-2) rectangle ++(1,-1);
}
\fill[fillA] (14,-2) rectangle ++(1,-1);
\foreach \x in {3,...,9}{
  \fill[panelcyan] (\x,-3) rectangle ++(1,-1);
}
\foreach \x in {10,...,15}{
  \fill[fillA] (\x,-3) rectangle ++(1,-1);
}
\foreach \x in {4,...,15}{
  \fill[fillA] (\x,-4) rectangle ++(1,-1);
}
\foreach \x in {5,...,11}{
  \fill[fillA] (\x,-5) rectangle ++(1,-1);
}
\foreach \x in {0,...,12}{
  \draw[cell] (\x,0) rectangle ++(1,-1);
}
\foreach \x in {1,...,13}{
  \draw[cell] (\x,-1) rectangle ++(1,-1);
}
\foreach \x in {2,...,14}{
  \draw[cell] (\x,-2) rectangle ++(1,-1);
}
\foreach \x in {3,...,15}{
  \draw[cell] (\x,-3) rectangle ++(1,-1);
}
\foreach \x in {4,...,16}{
  \draw[cell] (\x,-4) rectangle ++(1,-1);
}
\foreach \x in {5,...,17}{
  \draw[cell] (\x,-5) rectangle ++(1,-1);
}
\foreach \x in {3.5,4.5,5.5,6.5,7.5,8.5}{
  \draw[patharrow] (\x,-3.5)--++(1.1,0);
}

\draw[patharrow] (9.5,-3.5)--++(0,1.2);

\foreach \x in {9.5,10.5,11.5,12.5}{
  \draw[patharrow] (\x,-2.5)--++(1.1,0);
}
\draw[patharrow] (13.5,-2.5)--++(0,1.2);

\node[anchor=west,font=\Large] at (14.2,-1.0) {$\xi$};

\foreach \x in {-2.5}{
  \draw[patharrow] (2.5,\x)--++(0,1.2);
}

\foreach \x in {2.5,3.5,4.5}{
  \draw[patharrow] (\x,-1.5)--++(1.1,0);
}
\draw[patharrow] (5.45,-1.5)--++(0,1.1);
\foreach \x in {5.5,6.5,7.5,8.5,9.5,10.5,11.5}{
  \draw[patharrow] (\x,-0.5)--++(1.1,0);
}

\node[anchor=west,font=\Large] at (13.1,0.3) {$\eta$};

\foreach \x in {5.5,6.5,7.5,8.5,9.5,10.5}{
  \draw[patharrow] (\x,-5.5)--++(1.1,0);
}

\draw[patharrow] (11.5,-5.5)--++(0,1.2);

\foreach \x in {11.5,12.5,13.5,14.5}{
  \draw[patharrow] (\x,-4.5)--++(1.1,0);
}
\draw[patharrow] (15.5,-4.5)--++(0,1.2);

\node[anchor=west,font=\Large] at (16.2,-3.0) {$\xi$};

\fill[dotstyle] (6.4,-6.4) circle (1.1pt);
\fill[dotstyle] (6.7,-6.7) circle (1.1pt);
\fill[dotstyle] (7.0,-7.0) circle (1.1pt);
\fill[dotstyle] (11.4,-6.4) circle (1.1pt);
\fill[dotstyle] (11.7,-6.7) circle (1.1pt);
\fill[dotstyle] (12.0,-7.0) circle (1.1pt);
\fill[dotstyle] (18.4,-6.4) circle (1.1pt);
\fill[dotstyle] (18.7,-6.7) circle (1.1pt);
\fill[dotstyle] (19.0,-7.0) circle (1.1pt);

\end{scope}

\end{tikzpicture}
\caption{Panel a) shows in white and light blue the infinite diagram $Z_{13}$, while in blue only a sub diagram having for outline the path $\xi$. Panel b) shows two examples (one in light blue only and one in both light and dark blue) of diagrams in $D^{13}_{\eta,\xi}$.}
    \label{Fig: diagrams}
\end{figure}
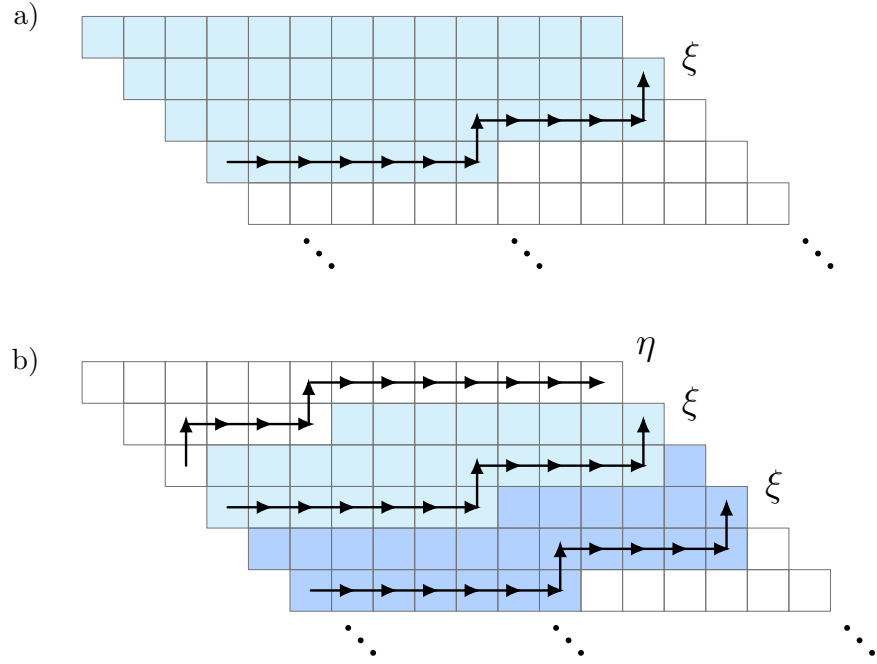

\section{Results}
\subsection{Walks in the TASEP state space}
In this section we map the problem of the enumeration of walks between two configurations in the open TASEP to that of the enumeration of standard Young tableaux for certain non-classical shapes. Surprisingly, the diagrams in question have been of recent interest in the combinatorial literature \cite{VeraLopez2017,EnumTableaux,periodiPPartitions}. This section extends to the open case the argument introduced in \cite{elizalde2022walks} linking the periodic TASEP and cylindric tableaux.

To do so we will work with the length $N$ open TASEP state space $X_N$, the set of length $N$ outlines $\{\uparrow,\to\}^N$ and finally the set of diagrams $D_{N+1}$. We introduce $h\colon D_{N+1}\to\{\uparrow,\to\}^N$ which maps a diagram to its outline, $g\colon \{\uparrow,\to\}^N\to X_N$ such that each upward step is mapped to an occupied site and each rightward step to an  empty one and $f=g \ \circ h\colon D_{N+1}\to X_N$.

Finally, given a standard Young tableau $T\in SYT(D)$, for any $k<|D|$ we define the subdiagram $T^k\subset D$ as $T^k=\cup_{i=1}^k T^{-1}(i)$.

\begin{theorem}
    Let $x,y\in X_N$ and $D_1, \ D_2$ be any pair of diagrams in $ D_{N+1}$ such that $D_1\subset D_2$, $h(D_1)=g^{-1}(x)$, $h(D_2)=g^{-1}(y)$ and $n=\left|D_2\setminus D_1\right|$. 
    If $\ W(x,y,n)$ is the 
    number of all $n$-step walks from $x$ to $y$, then:
    \[
    W(x,y,n)=|{SYT(D_2\setminus D_1})|.
    \]
     
\end{theorem}

\begin{proof}
    Let \(\mathrm{Y}_{N+1}\) be the directed Young graph whose vertices are the diagrams in \(D_{N+1}\), with an edge \(D\to D'\), with $D,\,D'\in D_N$, whenever \(D'\) is obtained from \(D\) by adding a corner. Let \(\mathrm{X}_N\) be the directed TASEP state graph with vertex set \(X_N\), where \(x\to y\) whenever \(y\) is obtained from \(x\) by one valid TASEP move. We claim that the map \(f=g\circ h\colon D_{N+1}\to X_N\) identifies the outgoing edges of each \(D\in D_{N+1}\) bijectively with the outgoing edges of \(f(D)\), that is, it is a covering map for the two graphs. Indeed, adding a corner to \(D\), if the corner does not happen to be the leftmost or rightmost cell in a row of $Z_{N+1}$, corresponds at the level of the outline \(h(D)\in\{\uparrow,\to\}^N\), to replacing a consecutive pattern \((\uparrow,\to)\) with \((\to,\uparrow)\) at a unique position. If the corner does happen to be a leftmost cell then the first step of the outline is changed from a rightward step to an upward one, vice-versa if it is a rightmost cell. Under the map \(g\), this is exactly the local TASEP update in which a particle moves one step to the right into an empty site. Conversely, every allowed TASEP move in \(f(D)\) arises in this way from a unique such exchange in the outline, hence from a unique corner addition in \(D\).

    Since \(f\) preserves directed adjacency and induces a bijection between the outgoing edges at every vertex, every directed walk
    \[
        D_1=D^{(0)}\to D^{(1)}\to\cdots\to D^{(n)}
    \]
    in \(\mathrm{Y}_{N+1}\) projects to a unique directed walk
    \[
        x=f(D_1)\to f(D^{(1)})\to\cdots\to f(D^{(n)})
    \]
    in \(\mathrm{X}_N\). It follows that \(W(x,y,n)\) is equal to the number of directed walks from \(D_1\) to \(D_2\) of length \(n\). Since each edge in \(\mathrm{Y}_{N+1}\) adds exactly one corner, such a walk produces a sequence of diagrams of the form:
    \[
        D_1=D^{(0)}\subset D^{(1)}\subset\cdots\subset D^{(n)}=D_2,
    \]
    with \(|D^{(r)}\setminus D_1|=r\) for all \(r\).

    Given such a sequence, label by \(r\) the unique cell of \(D^{(r)}\setminus D^{(r-1)}\). This produces a filling of \(D_2\setminus D_1\) with the labels \(1,\dots,n\). Because corners can only be added in an order compatible with the partial order of the diagram, the resulting filling is order-preserving, hence a standard Young tableau. Conversely, every tableau \(T\in SYT(D_2\setminus D_1)\) determines a diagram sequence chain by setting
    \[
        D^{(r)}\coloneqq D_1\cup T^r,\qquad r=0,\dots,n,
    \]
    thus the set of walks in $Y_{N+1}$ and $SYT(D_2\setminus D_1)$ are in bijection. 
    Therefore the number of \(n\)-step walks from \(D_1\) to \(D_2\) is \(|SYT(D_2\setminus D_1)|\), and:
    \[
        W(x,y,n)=|SYT(D_2\setminus D_1)|.
    \]
\end{proof}

The previous theorem shows that the recent problem of enumerating the number of Young tableaux for shifted strips and shifted shapes coincides with the enumeration of closed walks in the TASEP with open boundaries, an example of which can be seen in Fig. \ref{fig:placeholder}. In fact, whenever $x=y$ the diagrams in $D^{N+1}_{h(x),h(y)}$ are shifted shapes. In particular, if one considers all closed walks from the configuration associated to the completely empty lattice, they recover the parallelogrammic shapes discussed in \cite{VeraLopez2017,Sun2015EnumerationOS}.

We remark that $f^D$ is invariant under reflections along the coordinate bisectors of the plane, and in particular under the map
\begin{equation}
    \label{symmetry}
    s\colon  (i,j)\mapsto(-j,-i).
\end{equation}
Now consider a diagram $D\in D_{N+1}$, and let $\eta=h(D)$ and $x=g(\eta)\in X_N$. Define $y=g(h(s(D)))\in X_N$ as the configuration encoded by the reflected diagram $s(D)$. A direct check shows that
\begin{equation*}
    y_i=\begin{cases}
        1 & \text{if } x_{N-i+1}=0,\\
        0 & \text{if } x_{N-i+1}=1,
    \end{cases}
\end{equation*}
that is, the invariance of the number of Young tableaux under \eqref{symmetry} is reflected in the well-known particle-hole mapping of the TASEP dynamics.

The results presented in \cite{periodiPPartitions} thus imply the non-trivial fact that the number of closed walks in the TASEP obeys a constant coefficient recurrence relation.

\begin{figure}
    \centering
\begin{tikzpicture}[x=0.55cm,y=0.55cm,>=Latex,line join=round,line cap=round]
\tikzset{
  cell/.style={draw=black!55, line width=0.35pt},
  patharrow/.style={-{Latex[length=2.2mm,width=1.8mm]}, line width=0.9pt},
  dotstyle/.style={fill=black,draw=none},
}
\colorlet{panelcyan}{cyan!14}

\begin{scope}[shift={(-0.5,-1.5)}]

\node[anchor=east,font=\fontsize{11}{11}\selectfont] at (0.5,1) {a)};

\foreach \x in {1.5,3,4.5}{
  \fill[panelcyan] (\x,.5) rectangle ++(1.5,-1.5);
}
\foreach \x in {3,4.5}{
  \fill[panelcyan] (\x,-1) rectangle ++(1.5,-1.5);
}
\foreach \x in {1.5,3,4.5}{
  \draw[cell] (\x,.5) rectangle ++(1.5,-1.5);
}
\foreach \x in {3,4.5,6}{
  \draw[cell] (\x,-1) rectangle ++(1.5,-1.5);
}
\foreach \x in {4.5,6,7.5}{
  \draw[cell] (\x,-2.5) rectangle ++(1.5,-1.5);
}
\foreach \x in {6,7.5,9}{
  \draw[cell] (\x,-4) rectangle ++(1.5,-1.5);
}
\node[anchor=east,font=\fontsize{11}{11}\selectfont] at (7,-1.75) {2};
\node[anchor=east,font=\fontsize{11}{11}\selectfont] at (5.5,-3.25) {1};
\node[anchor=east,font=\fontsize{11}{11}\selectfont] at (7,-3.25) {3};
\node[anchor=east,font=\fontsize{11}{11}\selectfont] at (7,-4.75) {5};
\node[anchor=east,font=\fontsize{11}{11}\selectfont] at (8.5,-3.25) {4};
\node[anchor=east,font=\fontsize{11}{11}\selectfont] at (4
(8.5,-4.75) {6};
\node[anchor=east,font=\fontsize{11}{11}\selectfont] at (10,-4.75) {7};

\end{scope}

\begin{scope}[shift={(8,-4.5)}]

\node[anchor=east,font=\fontsize{11}{11}\selectfont] at (4,4) {b)};

  \def\r{3mm}           
  \def\dx{3.5mm}          
  \def\sh{6.5mm}       

  \tikzset{
    circleNode/.style={draw,circle,inner sep=0, minimum size={2*\r}},
    lab/.style={font=\small}
  }

  \coordinate (L) at (5.5,0);
  \coordinate (T) at ( 8.5,3.5);
  \coordinate (R) at ( 11.5,0);
  \coordinate (B) at ( 8.5,-3.5);

  \node[circleNode,fill=white] (L1) at ($(L)+(-\dx,0)$) {};
  \node[circleNode,fill=white] (L2) at ($(L)+(\dx,0)$) {};

  \node[circleNode,fill=black] (T1) at ($(T)+(-\dx,0)$) {};
  \node[circleNode,fill=white] (T2) at ($(T)+(\dx,0)$) {};

  \node[circleNode,fill=black] (R1) at ($(R)+(-\dx,0)$) {};
  \node[circleNode,fill=black] (R2) at ($(R)+(\dx,0)$) {};

  \node[circleNode,fill=white] (B1) at ($(B)+(-\dx,0)$) {};
  \node[circleNode,fill=black] (B2) at ($(B)+(\dx,0)$) {};

  \draw[->,bend left=20,shorten >=\sh,shorten <=\sh]
    (L) to node[lab,above left]  {} (T);
  \draw[red,->,bend right=20,shorten >=\sh,shorten <=\sh]
    (R) to node[lab,above right] {2, 5}  (T);
  \draw[red,->,bend right=20,shorten >=\sh,shorten <=\sh]
    (B) to node[lab,below right] {1,4} (R);
  \draw[red, ->,bend left=20,shorten >=\sh,shorten <=\sh]
    (B) to node[lab,below left]  {7}  (L);

  \draw[red, ->,shorten >=\sh,shorten <=\sh]
    (T) -- node[lab,right] {3, 6} (B);
\end{scope}

\end{tikzpicture}
    \caption{Panel a) shows a standard Young tableau associated with a diagram of the form $D_2\setminus D_1$, where $D_2$ (represented as blue and white cells) and $D_1$ (represented as only blue cells) are both in $D_3$. Panel b) represents the path in the 2-sites TASEP configuration space in bijection with the standard Young tableaux of panel a).}
    \label{fig:placeholder}
\end{figure}
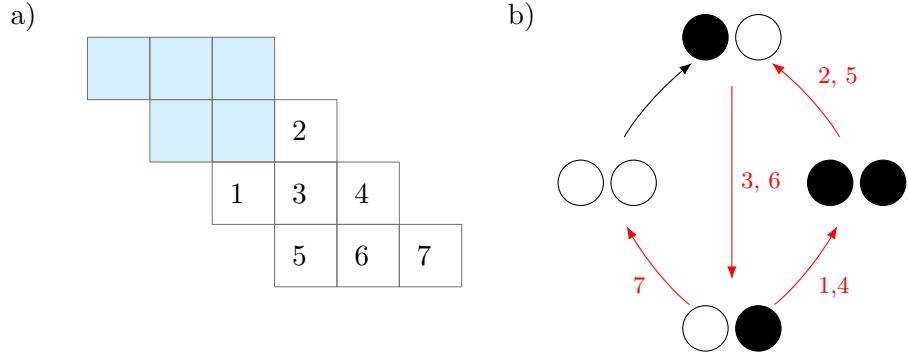

\subsection{The transition probability matrix}

The goal of the following subsection is to provide an order theoretic interpretation of the transition probability matrix $P(t)=e^{Ht}$ for the homogeneous TASEP, that is to write its entries as the sum of the exponential generating functions of certain combinatorial objects.

To do so we will interpret the power $H^n$ as a sum of length $n$ words in the alphabet $\Sigma=\{h_1,\dots,h_{N+1}\}$, whose set we denote as $\Sigma^n$. Thus, we write:
\begin{equation*}
H^n=\left(h_1+\dots+h_{N+1}\right)^n=\sum_{s\in\Sigma^n}s.
\end{equation*}
In this context a word is both an element in $\Sigma^n$ and the corresponding matrix obtained by computing the product it defines. 

The sitewise generators in \(\Sigma\) satisfy the relations \eqref{algebraic rel}. Hence, given a word $s$ in $\Sigma^n$ we call, for some $k\leq n$, $\overline{s}\in\Sigma^k$ the reduced word of $s$ if it constitutes a factorization of minimal length, up to a sign, of $s$. 

A reduced word is not unique, as its factors can commute according to $\eqref{algebraic rel}$. In general, we say that two words $s_1,\,s_2$ in $\Sigma$ are equivalent if it is possible to rewrite one as the other respecting the commutation relation and we write $s_1=_{\Sigma}\, s_2$.

We will now show that it is possible to identify any given equivalence class of words in $\Sigma^n$ with a labeled poset $\boldsymbol{D}=(D,\Sigma,\ell_D)$, where $D$ is a subdiagram of $Z_{N+1}$  defined in \eqref{inifinite shape}, $\Sigma$ is the local generator alphabet and $\ell_D\colon D\to\Sigma$ is the restriction to D of the  labeling function
\begin{equation*}
    \overline{\ell}\colon Z_{N+1}\to\Sigma,
\end{equation*}
defined as
\begin{equation}
    \label{labeling}
    \overline{\ell}\colon (i+k,i)\mapsto h_k.
\end{equation}

Given a reduced word $\overline{s}=s_1\cdots s_k$ in $\Sigma^k$ consider the function $\pi_s\colon [k]\to\Sigma$ such that $\pi_s(i)$ gives the generator appearing in the $i$-th position in $s$. We define $\preceq_s$ as the partial order on $[k]$ such that for all $i,\ j\in[k]$ one has  
\begin{equation*}
    \text{if}\quad[\pi(i),\pi(j)]\neq0 \quad \text{and}\quad j\geq i \quad\text{then}\quad j\preceq_s i \ .
\end{equation*}
The poset $P=\big([k],\preceq_s\big)$ inherits its cover relations from the commutation rules \eqref{algebraic rel}, thus it can be seen that it is in bijection with a sub-diagram of $Z_{N+1}$. Then the word $s$ lies in an equivalence class identified with $\boldsymbol{D}$, as $\boldsymbol{D}$ depends only upon the relative order in which the letters of $s$ appear and their commutation relations.

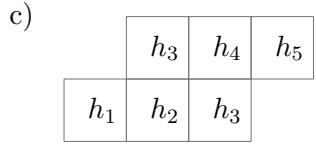
\begin{figure}
    \centering
\begin{tikzpicture}[x=0.55cm,y=0.55cm,>=Latex,line join=round,line cap=round]
\tikzset{
  cell/.style={draw=black!55, line width=0.35pt},
  patharrow/.style={-{Latex[length=2.2mm,width=1.8mm]}, line width=0.9pt},
  dotstyle/.style={fill=black,draw=none},
}
\begin{scope}[shift={(-0,-0.5)}]

\node[anchor=east,font=\fontsize{11}{11}\selectfont] at (-1,0) {a)};
\node[anchor=east,font=\fontsize{11}{11}\selectfont] at (8,0) {$s\ = \ h_3\,h_3\,h_5\,h_4\,h_2\,h_1\,h_4\,h_3\,h_1$};
\end{scope}
\begin{scope}[shift={(11,-0.5)}]

\node[anchor=east,font=\fontsize{11}{11}\selectfont] at (-0.8,0) {b)};
\node[anchor=east,font=\fontsize{11}{11}\selectfont] at (6,0) {$\overline{s}\ = \ h_3\,h_5\,h_4\,h_2\,h_3\,h_1$};

\end{scope}
\begin{scope}[shift={(6,-2.5)}]

\node[anchor=east,font=\fontsize{11}{11}\selectfont] at (-0,0) {c)};

\foreach \x in {2,3.5,5}{
  \draw[cell] (\x,0) rectangle ++(1.5,-1.5);
}
\foreach \x in {0.5,2,3.5}{
  \draw[cell] (\x,-1.5) rectangle ++(1.5,-1.5);
}
\node[anchor=east,font=\fontsize{11}{11}\selectfont] at (3.5,-0.75) {$h_3$};
\node[anchor=east,font=\fontsize{11}{11}\selectfont] at (5,-0.75) {$h_4$};
\node[anchor=east,font=\fontsize{11}{11}\selectfont] at (6.5,-0.75) {$h_5$};
\node[anchor=east,font=\fontsize{11}{11}\selectfont] at (2,-2.25) {$h_1$};
\node[anchor=east,font=\fontsize{11}{11}\selectfont] at (3.5,-2.25) {$h_2$};
\node[anchor=east,font=\fontsize{11}{11}\selectfont] at (5,-2.25) {$h_3$};

\end{scope}

\end{tikzpicture}

  \caption{Panel a) shows an example of a word $s$ in $\Sigma^9$, where $\Sigma=\{h_i\}_{i=1}^5$, b) one of its possible reduced words and c) the labeled poset associated with the equivalence class of $s$. }
    \label{fig: words-diagram-bijection}
\end{figure}

Finally, given a configuration $x\in {X}_N$, we denote by $\mathcal{N}^+(x)\subset X_N$ the set of configurations obtained by moving, independently and simultaneously, any subset (possibly empty) of particles of $x$ that can jump one site to the right (i.e.\ particles at sites $i$ such that $x_i=1$ and $x_{i+1}=0$). In other words, $y\in\mathcal{N}^+(x)$ if it can be obtained from $x$ by advancing any number of movable particles by one site. Analogously, we define $\mathcal{N}^-(x)\subset X_N$ as the set of configurations obtained by moving, independently and simultaneously, any subset (possibly empty) of particles that can jump one site to the left.

We are now ready to give the following:

\begin{lemma}
Let D be a diagram in $D_{\xi,\eta}^{N+1}$, where $\xi,\eta\in\{\uparrow,\rightarrow\}^N$, be such that there exists $D_0\in D_{\xi,\eta}^{N+1}$ with $D_0\subset D$. Let $s\in\Sigma^n$ be a word in the equivalence class induced by $(D,\Sigma,\ell_D)$, where $\ell_D$ is the restriction of $\overline{\ell}$ to $D$. Set $x=g(\xi)\in X_N$ and $y=g(\eta)\in X_N$. Then
\[
(s)_{ij}=
\begin{cases}
(-1)^{\,n-|D|+\kappa(z,y)}
& \text{if } j=\operatorname{idx}(x)\ \text{and}\ i=\operatorname{idx}(z)\ \text{for some } z\in \mathcal{N}^-(y),\\[4pt]
0 & \text{otherwise},
\end{cases}
\]
where $\kappa(z,y)$ denotes the number of particles that are moved one site to the left in passing from $y$ to $z$, that is, $\kappa(z,y)=\frac12\|y-z\|_1$).
\end{lemma}

\begin{proof}

We decompose each local generator $h_r\in\Sigma$ as the sum of two matrices, one containing its diagonal component and the other its off diagonal one:
\begin{equation}
    \label{expansion}
    h_r=-h_r^{(1)}+h_r^{(2)},
\end{equation}
where for $1<r<N+1$ we have
\[
h_r^{(1)}=I^{\otimes (r-2)}\otimes(E_{22}\otimes E_{11})\otimes I^{\otimes (N-r)},
\qquad
h_r^{(2)}=I^{\otimes (r-2)}\otimes(E_{12}\otimes E_{21})\otimes I^{\otimes (N-r)},
\]
while at the boundaries
\[
h_{1}^{(1)}=E_{11}\otimes I^{\otimes (N-1)},\qquad
h_{1}^{(2)}=E_{12}\otimes I^{\otimes (N-1)},
\]
and
\[
h_{N+1}^{(1)}=I^{\otimes (N-1)}\otimes E_{22},\qquad
h_{N+1}^{(2)}=I^{\otimes (N-1)}\otimes E_{12}.
\]

Let $\overline{s}$ be a reduced word for $s$. Since each cancellation $h_r^2=-h_r$ removes one letter and contributes one factor $-1$, we may write
\[
\overline{s}=(-1)^{n-m}\prod_{q=1}^{m}\overline{s}_q,
\qquad m=|D|.
\]
Expanding every factor according to \eqref{expansion} gives
\begin{equation}
\label{eq:expansion_lemma}
\overline{s}=(-1)^{n-m}\prod_{q=1}^{m}\bigl(\overline{s}_q^{(2)}-\overline{s}_q^{(1)}\bigr).
\end{equation}
Thus $s$ is a signed sum of $2^m$ Kronecker products, one for each choice of diagonal/off-diagonal contribution at every factor of the reduced word.

We claim that a summand in \eqref{eq:expansion_lemma} is nonzero if and only if the position, relative to the partial order induced by $\overline{s}$, of every factor chosen as a diagonal part corresponds to a corner of the diagram $D$. Indeed, if a diagonal term is chosen at a cell that is not a corner, then one of its neighboring factors in the reduced word leads to a null product, and the resulting term in the expansion vanishes. By contrast, at any corner the diagonal choice is compatible with the surrounding factors and therefore may contribute a non-null term.

Consider first the summand obtained by choosing the off-diagonal part at every factor:
\[
s_0\coloneqq (-1)^{n-m}\prod_{q=1}^{m}\overline{s}_q^{(2)}
      \coloneqq(-1)^{n-m}\bigotimes_{r=1}^{N}c_r.
\]
The factor $c_r$ is determined by the $r$-th step of the two outlines $\xi$ and $\eta$:
\begin{equation}
\label{eq:ci_from_outlines}
c_r=
\begin{cases}
E_{11} & \text{if } \xi_r=\eta_r=\rightarrow,\\
E_{22} & \text{if } \xi_r=\eta_r=\uparrow,\\
E_{12} & \text{if } \xi_r=\rightarrow,\ \eta_r=\uparrow,\\
E_{21} & \text{if } \xi_r=\uparrow,\ \eta_r=\rightarrow.
\end{cases}
\end{equation}
On the other hand, the basis matrix $E_{\operatorname{idx}(x),\operatorname{idx}(y)}$ admits the factorization
\begin{equation}
\label{eq:ci_from_configs}
E_{\operatorname{idx}(y),\operatorname{idx}(x)}=\bigotimes_{r=1}^{N}d_r,
\qquad
d_r=
\begin{cases}
E_{11} & \text{if } x_r=y_r=0,\\
E_{22} & \text{if } x_r=y_r=1,\\
E_{12} & \text{if } x_r=0,\ y_r=1,\\
E_{21} & \text{if } x_r=1,\ y_r=0.
\end{cases}
\end{equation}
Since the map $g$ identifies $\uparrow$ with $1$ and $\rightarrow$ with $0$, the descriptions in \eqref{eq:ci_from_outlines} and \eqref{eq:ci_from_configs} coincide. Therefore
\[
s_0=(-1)^{n-m}E_{\operatorname{idx}(y),\ \operatorname{idx}(x)}.
\]

Now choose a nonempty set $C$ of corners of $D$, and in \eqref{eq:expansion_lemma} select the diagonal part exactly at the factors whose relative order is given by the cells in $C$ and call $s_1$ the corresponding term in the expansion \eqref{eq:expansion_lemma}. Because factors whose relative position is given by corners commute with each other, the corresponding summand is again a nonzero Kronecker 
product with a single nonzero entry. For the same reason, the Kronecker factorization of $s_1$ corresponds to that of $s_0$ apart from the factors corresponding to the cells in $C$. Let $i,\,i+1$ be the indices associated to a corner in $C$ and $b_i,\, b_{i+1}$ the respective Kronecker factors of $s_1$, then:
 \begin{equation*}
    b_{i}=\left\{
        \begin{array}{lll}
         E_{12}  &\text{if }  &\eta_{i}=\xi_{i}=\uparrow, \\ 
         E_{11}  &\text{if }  &\eta_{i}=\uparrow, \ \xi_{i}=\rightarrow,
        \end{array}\right.
\end{equation*}
and
\begin{equation*} 
    b_{i+1}=\left\{
        \begin{array}{lll}
         E_{21}  &\text{if }  &\eta_{i+1}=\xi_{i+1}=\rightarrow, \\ 
         E_{22}  &\text{if }  &\eta_{i+1}=\rightarrow, \ \xi_{i+1}=\uparrow,
        \end{array}\right.
    \end{equation*}

Thus, relative to $s_0$, each chosen corner replaces one allowed right jump by a "stay" at the corresponding particle position; on the level of configurations this means that the arrival configuration $y$ is modified by moving that particle one site to the left. Hence the resulting term is
\[
(-1)^{n-m+|C|}E_{\operatorname{idx}(z_C),\operatorname{idx}(x)},
\]
where $z_C\in\mathcal{N}^-(y)$ is obtained from $y$ by moving to the left exactly the particles corresponding to the corners in $C$. Observing that $\kappa(z_C,y)=|C|$, then every nonzero summand of \eqref{eq:expansion_lemma} is of the form
\[
(-1)^{n-|D|+\kappa(z,y)}E_{\operatorname{idx}(z),\operatorname{idx}(x)}
\qquad\text{for some }z\in\mathcal{N}^-(y).
\]

Collecting all summands proves that the only nonzero entries of $s$ occur in column $\operatorname{idx}(x)$ and in the rows indexed by configurations $z\in\mathcal{N}^-(y)$, with coefficient $(-1)^{n-|D|+\kappa(z,y)}$ one recovers the thesis.

\end{proof}

The previous lemma does not consider the case in which $D^{N+1}_{\xi, \eta}$ contains no proper subset of $D$, i.e. $D$ is the diagram of least cardinality. In this case, the paths $\xi,\, \eta$ in $\mathbb{Z}^2$ may partially overlap each other, if viewed as originating from the same point in the plane. Thus, $D$ may not identify a unique couple of outlines, but a set which is uniquely defined on a certain index set $I_D\subset[N]$, while on the complementary set $[N]\setminus I_D$ the two outline are only required to coincide. Considering again the restriction $\ell_D$ of $\overline{\ell}$ on $D$, it is immediate that its image is a proper subset of $\Sigma$, i.e. all words in the equivalence class set by $(D,\Sigma,\ell_D)$ do not contain the entire alphabet $\Sigma$.  The index set can be given as:
\begin{equation}
    \label{index set}
    \forall\,i\in\{2,\cdots,N\},\,i\in I_D\Leftrightarrow\,\exists d\in D \,\text{ s. t. }\,\ell_D(d)\in\{h_{i-1},h_i\}.
\end{equation}

For this reason, for any given configuration $x\in X_N$ and index set $I\in[N]$, we introduce the notion of sub-configuration $x_I$, being the configuration obtained by ignoring the lattice sites not in $I$. We will call $\mathcal{N}^-(x,\,I)$  the set of all other sub-configurations that can be reached by moving to the left into an empty site a particle in $x_i$, without reaching any site indexed by $[N]\setminus I$. 
Finally, we introduce the set of pairs of configurations $A(x,y,I)=\big\{(a,b)\in X_N^2\,\big|a_I=x_I, \ b_I=y_I, \ a_{[N]\setminus I}=b_{[N]\setminus I}\big\}$. We are now ready to give the following.
\begin{lemma}
    Let $D\in D^{N+1}_{\xi,\eta}$ be such that it does not admit any proper subset in $ D^{N+1}_{\xi,\eta}$, call
    $x=g(\xi)$ and $y=g(\eta)$.
    Let $s\in\Sigma^n$ be a word in the equivalence class induced by $(D,\Sigma,\ell_D)$ and $I_D\subset [N]$ the index set defined in \eqref{index set}.
    Then
    \[
    (s)_{ij}=
    \begin{cases}
    (-1)^{n-|D|-\kappa(q,z)}
    & \text{if } \ j=\operatorname{idx}(w) \text{, } \ i=\operatorname{idx}(z),\\
    0 & \text{otherwise},
    \end{cases}
    \]
for some $(w,q)\in A(x,y,I_D)$ and  $z\in \mathcal{N}^-(q,I_D)$, where  $\kappa(q,z)=\frac12\|z-q\|_1$. 
\end{lemma}
\begin{proof}

Let \(m=|D|\),  and let \(\overline{s}\) be a reduced word for \(s\). As in the previous lemma, each cancellation removes one letter and contributes a factor \(-1\), so that
\[
\overline{s}=(-1)^{n-m}\prod_{q=1}^{m}\overline{s}_q .
\]
Expanding each factor according to \eqref{expansion}, we obtain
\[
\overline{s}=(-1)^{n-m}\prod_{q=1}^{m}\bigl(\overline{s}_q^{(2)}-\overline{s}_q^{(1)}\bigr),
\]
hence \(s\) can be written as a signed sum of terms, each corresponding to a choice of diagonal or off-diagonal contribution for every factor of the reduced word.

The proof now follows as in the previous lemma. The only difference is that, since the image of \(\ell_D\) is a proper subset of \(\Sigma\), there exist indices \(r\in[N]\setminus I_D\) for which neither \(h_r\) nor \(h_{r+1}\) appears in \(\overline{s}\). For such indices, no constraint is imposed by the diagram \(D\), and therefore the corresponding Kronecker factors in every summand reduce to the identity matrix \(I=E_{11}+E_{22}\).

Thus, any nonzero summand in the expansion of \(s\) can be written as a Kronecker product
\[
S=(-1)^{n-m+\varepsilon(S)}\bigotimes_{r=1}^{N}c_r,
\]
where \(\varepsilon(S)\) is the number of diagonal choices. For indices \(r\in I_D\), the factors \(c_r\) are exactly as in the previous lemma and are determined by the pair of outlines \(\xi,\eta\), hence by the configurations \(x=g(\xi)\) and \(y=g(\eta)\). For indices \(r\notin I_D\), the factors \(c_r\) are either \(E_{11}\) or \(E_{22}\), and therefore enforce that the corresponding sites of the initial and final configurations coincide. This follows from the identification established in the previous lemma, where the matrices \(E_{11}\) and \(E_{22}\) correspond respectively to the cases in which both outlines have the same step at position \(r\), namely \((\rightarrow,\rightarrow)\) and \((\uparrow,\uparrow)\). Under the map \(g\), these correspond to having equal occupation variables at site \(r\) in both configurations, thus no change occurs at that site.

It follows that any nonzero contribution to \(s\) must be supported on pairs of configurations \((w,q)\) such that \((w,q)\in A(x,y,I_D)\), that is, \(w\) and \(q\) agree outside \(I_D\), while their restriction to \(I_D\) coincides with \(x\) and \(y\), respectively. On the sites indexed by \(I_D\), the same reasoning as in the previous lemma applies: the term obtained by selecting all off-diagonal contributions produces a matrix unit
\[
(-1)^{n-m}E_{\operatorname{idx}(q),\operatorname{idx}(w)},
\]
while each admissible diagonal choice corresponds to replacing one allowed right jump by a stay, which amounts to moving one particle one site to the left without leaving the index set \(I_D\).

Therefore, starting from a configuration \(q\) such that \((w,q)\in A(x,y,I_D)\), any sequence of admissible diagonal choices produces a configuration \(z\in\mathcal{N}^-(q,I_D)\), and the corresponding summand is
\[
(-1)^{n-m+\kappa(q,z)}E_{\operatorname{idx}(z),\operatorname{idx}(w)},
\]
where \(\kappa(q,z)=\frac12\|q-z\|_1\). Conversely, every configuration \(z\in\mathcal{N}^-(q,I_D)\) arises in this way.

Collecting all contributions, we conclude that the only nonzero entries of \(s\) are those for which
\[
j=\operatorname{idx}(w),\qquad i=\operatorname{idx}(z),
\]
with \((w,q)\in A(x,y,I_D)\) and \(z\in\mathcal{N}^-(q,I_D)\), and in that case the coefficient is
\[
(-1)^{n-|D|-\kappa(q,z)}
\]
which gives the thesis.

\end{proof}
With the previous two lemmas we are able to compute the contribution to the transition probability matrix of any term $s\in\Sigma^n$ belonging to an equivalence class identified by a labeled poset $(D,\Sigma,\ell_D)$, where $D$ is a diagram in $D^{N+1}_{\xi,\eta}$ for some paths $\eta,\,\xi\in\{\rightarrow,\uparrow\}^N$ and $\ell_D$ is defined as the restriction over $D$ of \eqref{labeling}. In the following we show how such elements constitute all the non zero terms in the expansion of $H_N^n$ for any $n$, and we provide an order theoretic description of the transition probability matrix. 

\begin{theorem}
    \label{prob matrix theorem}
    Let $P(t)=e^{H_N\,t}\in\mathbb{R}^{2^N\times 2^N}$ be the transition probability matrix of the N-sites open TASEP, let $x,y\in X_N$ two configurations, and set $i=\operatorname{idx}(y)$ and $j=\operatorname{idx}(x)$. Then
    \[\big(P(t)\big)_{ij}=\sum_{z\in\mathcal{N}^+(y)}\sum_{\ \ \ D\in D^{N+1}_{\xi,\eta}} (-1)^{|D|+\kappa(y,z)}\mathcal{F}^D(-t),
    \quad \text{for every $i,j\in [2^N]$,}
    \]
    where $\mathcal{F}^{D}(t)$ is the exponential generating function defined in \eqref{egf}, $\eta=g^{-1}(z)$, $\xi=g^{-1}(x)$, and $\kappa(w,z)=\frac12\|z-w\|_1$\,.
\end{theorem}
\begin{proof}
    We start by proving that any non-zero word $s\in\Sigma^n$ lies in an equivalence class identified by a labeled poset of the form:
    \begin{equation}
    \label{labelled poset}
    \boldsymbol{D}=(D,\Sigma,\ell_D),
    \end{equation}
    where $D$ is a diagram in $D^{N+1}_{\xi,\eta}$ for some paths $\xi,\,\eta\in\{\uparrow,\rightarrow\}^N$ and $\ell_D$ is as in \eqref{labeling}.
    
    Take a generic triplet $\boldsymbol{P}=(P,\Sigma,\ell)$, where $P$ is a poset, $\Sigma$ is the usual alphabet set, and $\ell$ is a labeling function $\ell\colon P\to\Sigma$. Consider two elements $a,\, b\in P$ such that $\ell(a)=\ell(b)=h_i$ for $i\in\{2,\dots,N\}$ and for all elements $c\in P$ for which $a\prec c\prec b$ there does not exist any $c'$ such that $\ell(c')=h_i$.

    Due to the commutation relations in \eqref{algebraic rel} one has that $\ell(c')\in\{h_{i-1},h_{i+1}\}$ for all $c'$ such that $a\prec c\prec b$  and, because of the cancellation rules, there can be at most two such elements in $P$. If only one factor was present, then any reduced word induced by $\boldsymbol{P}$ would contain a factor of the form $h_i \, h_{i\pm 1}\, h_i$ and thus be null due to \eqref{algebraic rel}. If one repeats the same process for all other labels then it recovers that $\boldsymbol{P}$ is either of the form \eqref{labelled poset} or the matrices in its equivalence class are null.
    
    We now wish to compute the contribution of all such words to $\big( P(t)\big)_{ij}$, we start by noting that
    \begin{equation*}
        \big(P(t)\big)_{ij}=\sum_{n=0}^\infty \left(H^n_N\right)_{ij} \frac{t^n}{n!}=\sum_{n=0}^\infty\sum_D(-1)^{n-|D|+\kappa(y,q)}\,f^D(n)\frac{t^n}{n!},
    \end{equation*}
    where the sum over $D$ runs over all diagrams for which the matrices in the equivalence class of $(D,\Sigma,\ell)$ have a non-zero entry in $(i,j)$, and $q\in X_N$ is such that $f(D)=q$. We can now characterize all such diagrams and their contribution through Lemma~1 and Lemma~2, obtaining
    \begin{equation}
        \label{prob exp}
        \big(P(t)\big)_{ij}=\sum_{n=0}^\infty\,\sum_{z\in\mathcal{N}^+(y)}\sum_{\ \,  D\in D^{N+1}_{\xi,\eta}}(-1)^{|D|+\kappa(y,z)}f^D(n)\frac{(-t)^n}{n!},
    \end{equation}
    where $\eta=g^{-1}(z)$, $\xi=g^{-1}(x)$. By setting
    \begin{equation*}
        a_{nk}=(-1)^{|D|+\kappa(y,z)}\frac{f^D(n)}{n!},
    \end{equation*}
    where $k$ is an index running on the diagrams appearing in the inner sums, equation \eqref{prob exp} becomes
    \begin{equation}
        \label{double sum}
        \big(P(t)\big)_{ij}=\sum_{n=0}^{\infty}\sum_{k=0}^{\infty} a_{nk}\,(-t)^n=\sum_{n=0}^{\infty}\sum_{k=0}^{K_n} a_{nk}\,(-t)^n,
    \end{equation}
    where the last equality is due to the fact that the inner sum is finite for any choice of $n$, as $f^{D}(n)=0$ whenever $n<|D|$. Thus, for any $n$ there must exist a finite $K_n$ for which $a_{nk}$ is zero for all $k>K_n$.
    
    We now focus on inner sum. We have:
    \begin{equation*}
        \sum_{k=0}^{\infty} |a_{nk}|=\sum_{z\in\mathcal{N}^+(y)}\sum_{\ \,  D\in D^{N+1}_{\xi,\eta}}\frac{f^D(n)}{n!}<\frac{(N+1)^n}{n!},
    \end{equation*}
    where the inequality is due to the fact that the sum is counting the words in $\Sigma^n$ subject to some restriction upon the ordering of their letters, thus is bound to be less than $|\Sigma^n|=(N+1)^n$. Then of course the Eq. \eqref{double sum} is bounded from above by 	$e^{(N+1)|t|}$ for all $t$ and the sum is absolutely convergent.
    
    Now, being that a series is absolutely convergent if and only if it is \textit{unconditionally convergent} \cite[Theorem~3.56]{rudin1976principles}, we  rearrange the sum in \eqref{prob exp} as
    \begin{equation}    
    \label{final sum}
    \begin{split}
    \big(P(t)\big)_{ij}&=\sum_{z\in\mathcal{N}^+(y)}\sum_{\ \,  D\in D^{N+1}_{\xi,\eta}}(-1)^{|D|+\kappa(y,z)}\sum_{n=0}^\infty\,f^D(n)\frac{(-t)^n}{n!}\\
    &=\sum_{z\in\mathcal{N}^+(y)}\sum_{\ \, D\in D^{N+1}_{\xi,\eta}} (-1)^{|D|+\kappa(y,z)}\mathcal{F}^D(-t),
    \end{split}
    \end{equation}
     where we recall the definition of $\mathcal{F}^D(t)$ in Section~\ref{comb_prel}. The proof is concluded.
\end{proof}

\section{Conclusions}

In the present work we have investigated combinatorial mappings associated with the
TASEP on a finite lattice with open boundaries.

First, we showed that the number of transition sequences of prescribed length between
two configurations can be identified with the number of standard Young tableaux of
suitable non-classical diagrams determined by the corresponding outlines and by the
length of the walk. This extends to the open-boundary setting the correspondence
previously established for the TASEP with periodic boundary conditions
\cite{elizalde2022walks}. In the case of closed walks, and in particular of closed
walks based at the empty configuration, the construction recovers families of diagrams
that have recently attracted interest in enumerative combinatorics
\cite{periodiPPartitions,Sun2015EnumerationOS,VeraLopez2017}.

Subsequently we have provided a combinatorial interpretation to the finite-time transition probabilities of the system, viewing them as a signed sum of exponential generating functions of objects related to Young diagrams. This
complements the combinatorial interpretations of the steady-state distribution discussed
in \cite{Blythe_2007} by giving an analogous perspective on transient probabilities.

Several questions remain open. First, to our knowledge, no explicit enumeration formulas
are currently available for these exponential generating functions, except in elementary
cases. Second, it remains to be understood whether the recurrence and asymptotic results
known for the number \(f_D\) of standard Young tableaux of shifted strips and shifted
shapes extend to the generalized numbers \(f_D(n)\) introduced here.

We hope that the correspondence developed in this work will encourage further dialogue
between the theory of interacting particle systems and contemporary problems in
enumerative combinatorics, and that it may help clarify the structure of finite-time
transition probabilities for the TASEP on finite lattices with open boundaries.


\section*{Acknowledgments}
I sincerely thank Professors Marco Morandotti, Alessandro Pelizzola, and Marco Pretti for the numerous discussions and suggestions that made this work possible.


\bibliographystyle{plain}
\bibliography{sample}

\end{document}